\pdfoutput=1
\documentclass[11pt]{article}
\usepackage{subfig}
\usepackage{amssymb,amsmath,amsthm}
\usepackage{graphics,epsfig}
\usepackage{algorithmic}
\usepackage{algorithm}
\usepackage[margin=1 in]{geometry}

\usepackage{lipsum}
\usepackage{mathtools}
\usepackage{cuted}
\usepackage{hhline}

\newtheorem{lem}{Lemma}[section]

\newtheorem{thm}[lem]{Theorem}

\title{Improved Approximation Algorithms for Minimizing the Total Weighted Completion Time of Coflows}
\author{Chi-Yeh~Chen 
\\ Department of Computer Science and Information
Engineering, \\ National Cheng Kung University, \\
Taiwan, ROC. \\
chency@csie.ncku.edu.tw.}

\begin{document}

\maketitle
\begin{abstract}
This paper addresses the challenging scheduling problem of coflows with release times, with the objective of minimizing the total weighted completion time. Previous literature has predominantly concentrated on establishing the scheduling order of coflows. In advancing this research, we contribute by optimizing performance through the determination of the flow scheduling order. The proposed approximation algorithm achieves approximation ratios of $3$ and $2+\frac{1}{LB}$ for arbitrary and zero release times, respectively, where $LB$ is the minimum lower bound of coflow completion time. To further improve time complexity, we streamline linear programming by employing interval-indexed relaxation, thereby reducing the number of variables. As a result, for $\epsilon>0$, the approximation algorithm achieves approximation ratios of $3 + \epsilon$ and $2 + \epsilon$ for arbitrary and zero release times, respectively. Notably, these advancements surpass the previously best-known approximation ratios of 5 and 4 for arbitrary and zero release times, respectively, as established by Shafiee and Ghaderi.

\begin{keywords}
Scheduling algorithms, approximation algorithms, coflow, datacenter network.
\end{keywords}
\end{abstract}

\section{Introduction}\label{sec:introduction}
Large-scale data centers are a crucial component of cloud computing, and their advantages in application-aware network scheduling~\cite{Chowdhury2014, Chowdhury2015, Zhang2016, Agarwal2018} have been demonstrated, especially for distributed applications with structured traffic patterns such as MapReduce~\cite{Dean2008}, Hadoop~\cite{Shvachko2010, borthakur2007hadoop}, Dryad~\cite{isard2007dryad}, and Spark~\cite{zaharia2010spark}. Data-parallel applications typically alternate between the computing and communication stages. In the computing stage, these applications produce a significant amount of intermediate data (flows) that needs transmission across diverse machines for subsequent further in the communication stage. Data centers must possess robust data transmission and scheduling capabilities to address the massive scale of applications and their corresponding data transmission requirements. The communication patterns between the flows among these two sets of machines within the data center can be abstracted as coflow traffic~\cite{Chowdhury2012}. A coflow refers to a collection of interdependent flows, where the completion time of the entire group relies on the completion time of the last flow within the collection~\cite{shafiee2018improved}. 

In the research on coflow scheduling, it can be categorized into scheduling in electrical packet switches and scheduling in optical circuit switches~\cite{Zhang2021, Tan2021, Li2022}. An optical circuit switch exhibits a significantly higher data transfer rate but consumes less power. However, each ingress or egress port in an optical circuit switch is limited to establishing at most one circuit for data transmission at a time. The process of reconfiguring to establish new circuits in an optical circuit switch takes a fixed period, referred to as the reconfiguration delay~\cite{Zhang2021}. On the other hand, an electrical packet switch allows for the flexible allocation of bandwidth to each link, and flows can preempt one another. This paper addresses the coflow scheduling problem on packet switches, aiming to optimize the scheduling of coflows to minimize the total weighted coflow completion time.

\begin{table*}[!ht]
\caption{Theoretical Results}
\centering
\begin{tabular}{ccccc}
\hline
 Case & \multicolumn{2}{c}{Best known} & \multicolumn{2}{c}{This paper} \\ 
      & deterministic   &   randomized &  deterministic  &  randomized        \\ \hhline{=====}
 Without release times  &   4~\cite{shafiee2018improved}       &    $2e$~\cite{Shafiee2017}      &   $2 + \epsilon$        &  $2 + \epsilon$      \\
 With release times     &   5~\cite{shafiee2018improved}       &    $3e$~\cite{Shafiee2017}     &   $3 + \epsilon$        &   $3 + \epsilon$     \\ 
\hline
\end{tabular}
\label{tab:results}
\end{table*}

\subsection{Related Work}
The concept of coflow abstraction to characterize communication patterns within data centers is firstly introduced by Chowdhury and Stoica~\cite{Chowdhury2012}.
The coflow scheduling problem has been proven to be strongly $\mathcal{NP}$-hard. Since the concurrent open shop problem can be reduced to a coflow scheduling problem, 
it is $\mathcal{NP}$-hard to approximate the coflow scheduling problem within a factor better than $2-\epsilon$ due  to the inapproximability of the concurrent open shop problem~\cite{Bansal2010, Sachdeva2013}.

To minimize the weighted completion time of coflows, Qiu \textit{et al.}~\cite{Qiu2015} pioneered the development of algorithms that partition coflows into disjoint groups. Each group is then treated as a single coflow and assigned to a specific time interval. They obtained a deterministic approximation ratio of $\frac{64}{3}$ and a randomized approximation ratio of $(8+\frac{16\sqrt{2}}{3})$. When coflows are released at arbitrary times, their algorithms achieved a deterministic approximation ratio of $\frac{67}{3}$ and a randomized approximation ratio of $(9+\frac{16\sqrt{2}}{3})$. However, Ahmadi \textit{et al.}\cite{ahmadi2020scheduling} demonstrated that their deterministic approximation ratio is $\frac{76}{3}$ in the deterministic algorithm for coflow scheduling with release time.

Khuller \textit{et al.}~\cite{khuller2016brief} employ an approximate algorithm for the concurrent open shop problem to establish the scheduling order of coflows. This approach yields an approximation ratio of $12$ for coflow scheduling with arbitrary release times and an approximation ratio of $8$ when all coflows have zero release times. Additionally, they achieve a randomized approximation ratio of $3+2\sqrt{2} \approx 5.83$ when all coflows have zero release times. Shafiee and Ghaderi~\cite{Shafiee2017} further improved the randomized approximation ratios to $3e$ and $2e$ when coflows have arbitrary release times or no release times, respectively.

Shafiee and Ghaderi~\cite{shafiee2018improved} proposed the currently known best algorithm, achieving approximation ratios of 5 and 4 for arbitrary and zero release times, respectively. Moreover, Ahmadi \textit{et al.}~\cite{ahmadi2020scheduling} utilized a primal-dual algorithm to reduce time complexity and achieved identical results.

\subsection{Our Contributions}
This paper is dedicated to addressing the coflow scheduling problem and presents a variety of algorithms along with their corresponding results. The specific contributions of this study are outlined below:

\begin{itemize}
\item We present a randomized approximation algorithm with pseudo-polynomial complexity, achieving expected approximation ratios of $3$ and $2$ for the coflow scheduling problem with arbitrary and zero release times, respectively.

\item We present a deterministic approximation algorithm with pseudo-polynomial complexity that achieves approximation ratios of $3$ and $2+\frac{1}{LB}$ for the coflow scheduling problem with arbitrary and zero release times, respectively, where $LB$ is the minimum lower bound of coflow completion time.

\item We present a randomized approximation algorithm with polynomial complexity, achieving expected approximation ratios of $3 + \epsilon$ and $2 + \epsilon$ for the coflow scheduling problem with arbitrary and zero release times, respectively.

\item We present a deterministic approximation algorithm with polynomial complexity that achieves approximation ratios of $3 + \epsilon$ and $2 + \epsilon$ for the coflow scheduling problem with arbitrary and zero release times, respectively.
\end{itemize}

A summary of our theoretical findings is provided in Table~\ref{tab:results}.

\subsection{Organization}
The structure of this paper is outlined as follows. In Section~\ref{sec:Preliminaries}, we provide an introduction to fundamental notations and preliminary concepts that will be employed in subsequent sections. Section~\ref{sec:SOTA} offers an overview of our high-level ideas. The main algorithms are presented in the following sections: Section~\ref{sec:Algorithm1} introduces the randomized algorithm addressing the coflow scheduling problem, while Section~\ref{sec:Algorithm2} details the deterministic algorithm designed for the same problem. Section~\ref{sec:Algorithm3} proposes a randomized algorithm with interval-indexed linear programming relaxation to tackle the coflow scheduling problem, and Section~\ref{sec:Algorithm4} elaborates on the deterministic algorithm for the same. Finally, Section~\ref{sec:Conclusion} summarizes our findings and draws meaningful conclusions.

\section{Notation and Preliminaries}\label{sec:Preliminaries}
This paper considers a large-scale non-blocking switch with dimensions $N \times N$, where $N$ input links are connected to $N$ source servers, and $N$ output links are connected to $N$ destination servers. Let $\mathcal{I}=\left\{1,2,\ldots, N\right\}$ represent the set of source servers, and $\mathcal{J}=\left\{N+1, N+2,\ldots, 2N\right\}$ represent the set of destination servers. The switch can be viewed as a bipartite graph, with one side being $\mathcal{I}$ and the other side being $\mathcal{J}$. For simplicity, we assume that all links within the switch have the same capacity or speed.

Coflow is a collection of independent flows, and its completion time is determined by the completion time of the last flow in the collection. Let $D^{(k)}=\left(d_{i,j,k}\right)_{i,j=1}^{N}$ be the demand matrix for coflow $k$, where $d_{i,j,k}$ represents the size of the flow transmitted from input $i$ to output $j$ within coflow $k$. Since all links within the switch have the same capacity or speed, we can consider flow size as transmission time. Each flow is identified by a triple $(i, j, k)$, where $i \in \mathcal{I}$ represents the source node, $j \in \mathcal{J}$ represents the destination node, and $k$ corresponds to the coflow. Additionally, we assume that flows are composed of discrete integer units of data. For simplicity, we assume that all flows within the coflow are started simultaneously, as shown in~\cite{Qiu2015}.

Consider a set of coflows $\mathcal{K}$, each associated with a coflow release time $r_k$ and a positive weight $w_k$. In the context of coflow scheduling with release times, the objective of this paper is to determine a schedule that minimizes the total weighted completion time of the coflows, expressed as $\sum_{k\in \mathcal{K}} w_kC_k$. For ease of explanation, we assign different meanings to the same symbols with different subscript symbols. Subscript $i$ represents the index of the source (or the input port), subscript $j$ represents the index of the destination (or the output port), and subscript $k$ represents the index of the coflow. Let $\mathcal{F}_{i}$ be the collection of flows with source $i$, represented by $\mathcal{F}_{i}=\left\{(i, j, k)| d_{i,j,k}>0, \forall k\in \mathcal{K}, \forall j\in \mathcal{J} \right\}$, and $\mathcal{F}_{j}$ be the set of flows with destination $j$, given by $\mathcal{F}_{j}=\left\{(i, j, k)| d_{i,j,k}>0, \forall k\in \mathcal{K}, \forall i\in \mathcal{I} \right\}$. We also let $\mathcal{F}_{k}=\left\{(i, j, k)| d_{i,j,k}>0, \forall i\in \mathcal{I}, \forall j\in \mathcal{J} \right\}$ be the set of flows with coflow $k$. Let $\mathcal{F}$ be the collection of all flows, represented by $\mathcal{F}=\left\{(i, j, k)| d_{i,j,k}>0, \forall k\in \mathcal{K}, \forall i\in \mathcal{I}, \forall j\in \mathcal{J} \right\}$.

The notation and terminology used in this paper are summarized in Table~\ref{tab:notations}.

\begin{table}[ht]
\caption{Notation and Terminology}
    \centering
        \begin{tabular}{||c|p{5in}||}
    \hline
     $N$      & The number of input/output ports.         \\
    \hline
     $n$      & The number of coflows.         \\
    \hline
     $\mathcal{I}, \mathcal{J}$ & The source server set and the destination server set.         \\
    \hline    
     $\mathcal{K}$ & The set of coflows.         \\
    \hline
     $D^{(k)}$     & The demand matrix of coflow $k$. \\
    \hline    
     $d_{i,j,k}$     & The size of the flow to be transferred from input $i$ to output $j$ in coflow $k$.   \\
    \hline     
     $C_{i,j,k}$ & The completion time of flow $(i, j, k)$. \\
    \hline     
     $C_k$     & The completion time of coflow $k$.   \\
    \hline     
     $r_k$     & The released time of coflow $k$.  \\
    \hline     
     $w_{k}$   &  The weight of coflow $k$. \\
    \hline     
		 $\mathcal{F}$ & $\mathcal{F}=\left\{(i, j, k)| d_{i,j,k}>0, \forall k\in \mathcal{K}, \forall i\in \mathcal{I}, \forall j\in \mathcal{J} \right\}$ is the set of all flows. \\
    \hline     
		 $\mathcal{F}_{i}$ & $\mathcal{F}_{i}=\left\{(i, j, k)| d_{i,j,k}>0, \forall k\in \mathcal{K}, \forall j\in \mathcal{J} \right\}$ is the set of flows with source $i$. \\
		\hline 					
		 $\mathcal{F}_{j}$ & $\mathcal{F}_{j}=\left\{(i, j, k)| d_{i,j,k}>0, \forall k\in \mathcal{K}, \forall i\in \mathcal{I} \right\}$ is the set of flows with destination $j$. \\
		\hline 					
		 $\mathcal{F}_{k}$ & $\mathcal{F}_{k}=\left\{(i, j, k)| d_{i,j,k}>0, \forall i\in \mathcal{I}, \forall j\in \mathcal{J} \right\}$ is the set of flows with coflow $k$. \\
		\hline 				
		$T$, $\mathcal{T}$ & $T$ is the time horizon and $\mathcal{T}=\left\{0, 1, \ldots, T\right\}$ is the set of time indices. \\
		\hline 				
		$L_{ik}$, $L_{jk}$ & $L_{ik}=\sum_{j=1}^{N}d_{ijk}$ is the total amount of data that coflow $k$ needs to transmit through the input port $i$ and $L_{jk}=\sum_{i=1}^{N}d_{ijk}$ is the total amount of data that coflow $k$ needs to transmit through the output port $j$. \\
		\hline 						
		$LB$   & $LB=\min_k\max\left\{\max_{i}\left\{L_{ik}\right\},\max_{j}\left\{L_{jk}\right\}\right\}$ is the minimum lower bound of coflow completion time.\\
		\hline 						
		$L$, $\mathcal{L}$    & $L$ is the smallest value satisfying the inequality $(1+\eta)^L\geq T+1$ and $\mathcal{L}=\left\{0, 1, \ldots, L\right\}$ be the set of time interval indices. \\
		\hline 			
			$I_{\ell}$ & 		$I_{\ell}$ is the $\ell$th time interval where $I_{0}=[0, 1]$ and $I_{\ell}=((1+\eta)^{\ell-1},(1+\eta)^{\ell}]$ for $1\leq \ell \leq L$. \\
		\hline 			
		$|I_{\ell}|$ & $|I_{\ell}|$ is the length of the $\ell$th interval where $|I_{0}|=1$ and $|I_{\ell}|=\eta(1+\eta)^{\ell-1}$ for $1\leq \ell \leq L$. \\
		\hline  						
        \end{tabular}
    \label{tab:notations}
\end{table}

\section{Our High-level Ideas}\label{sec:SOTA}
In the context of the coflow scheduling problem with release times, previous literature has primarily focused on determining the scheduling order of coflows. This paper takes a further step by improving performance through the determination of flow scheduling order. The approach presented in this paper is inspired by Schulz and Skutella~\cite{Schulz2002}. Initially, we employ a linear programming to solve the scheduling problem. Subsequently, the solution obtained from linear programming is used to randomly determine the starting transmission times for flows. We use this time to determine the transmission order of flows and analyze the expected approximation ratio of this randomized algorithm. Next, we employ a deterministic algorithm obtained through the derandomized method. Finally, we use an interval-indexed linear programming relaxation to reduce the number of variables in linear programming.

In analyzing the completion time of coflows, we consider the flow $(i, j, k)$ that completes last within the coflow. In this particular flow, we need to account for all flows on ports $i$ and $j$ with a priority higher than $(i, j, k)$. We have three possible scenarios. The first scenario is when some flows on port $i$ and port $j$ can be transmitted simultaneously. The second scenario is when certain time intervals on ports $i$ and $j$ are simultaneously idle. The third scenario is when there is no simultaneous transmission on ports $i$ and $j$, and no time intervals are simultaneously idle. The first scenario is not the worst-case scenario. In the second scenario, flow $(i, j, k)$ can be transmitted during the idle time. The third scenario represents the worst-case scenario. Therefore, we analyze based on the third scenario.

\section{Randomized Approximation Algorithm for the Coflow Scheduling Problem}\label{sec:Algorithm1}
This section addresses the coflow scheduling problem. Let 
\begin{eqnarray*}
T=\max_{k\in \mathcal{K}} r_{k}+\max_{i\in \mathcal{I}}\sum_{k\in \mathcal{K}} L_{ik}+\max_{j\in \mathcal{J}}\sum_{k\in \mathcal{K}} L_{jk}-1.
\end{eqnarray*}
be the time horizon where $L_{ik}=\sum_{j=1}^{N}d_{ijk}$ is the total amount of data that coflow $k$ needs to transmit through the input port $i$ and $L_{jk}=\sum_{i=1}^{N}d_{ijk}$ is the total amount of data that coflow $k$ needs to transmit through the output port $j$. Let $\mathcal{T}=\left\{0, 1, \ldots, T\right\}$ be the set of time indices. For every flow $(i, j, k)\in \mathcal{F}$ and every time point $t=r_{k}, \ldots, T$, let $y_{ijkt}$ be the amount of time flow $(i, j, k)$ is transmitted within the time interval $(t, t+1]$. The problem can be formulated as a linear programming relaxation, given by:
\begin{subequations}\label{coflow:main}
\begin{align}
& \text{min}  && \sum_{k \in \mathcal{K}} w_{k} C_{k}     &   & \tag{\ref{coflow:main}} \\
& \text{s.t.} && \sum_{t=r_{k}}^{T} \frac{y_{ijkt}}{d_{ijk}} = 1, && \forall (i, j, k)\in \mathcal{F} \label{coflow:a} \\
&  && \sum_{(i, j, k) \in \mathcal{F}_{i}} y_{ijkt} \leq 1, && \forall i\in \mathcal{I}, \forall t\in \mathcal{T} \label{coflow:b} \\
&  && \sum_{(i, j, k) \in \mathcal{F}_{j}} y_{ijkt} \leq 1, && \forall j\in \mathcal{J}, \forall t\in \mathcal{T} \label{coflow:c} \\
&  && C_{k}\geq C_{ijk}, && \forall k\in \mathcal{K}, \forall (i, j, k)\in \mathcal{F}_{k}, \label{coflow:d} \\
&  && y_{ijkt} \geq 0, && \forall (i, j, k)\in \mathcal{F}, \forall t\in \mathcal{T}\label{coflow:g}
\end{align}
\end{subequations}
where
\begin{eqnarray}\label{coflow:d2}
C_{ijk} = \sum_{t=r_{k}}^{T}\left(\frac{y_{ijkt}}{d_{ijk}}\left(t+\frac{1}{2}\right)+\frac{1}{2}y_{ijkt}\right).
\end{eqnarray}
In the linear program (\ref{coflow:main}), the variable $C_{k}$ denotes the completion time of coflow $k$ in the schedule. Constraint (\ref{coflow:a}) ensures the fulfillment of the transmitting requirements for each flow. The switch capacity constraints (\ref{coflow:b}) and (\ref{coflow:c}) stipulate that the switch can transmit only one flow at a time in each input port $i$ and output port $j$. Constraint (\ref{coflow:d}) specifies that the completion time of coflow $k$ is constrained by the completion times of all its constituent flows. Additionally, the equation (\ref{coflow:d2}) represents the lower bound of the completion time for flow $(i, j, k)$, and this lower bound occurs during continuous transmission between $C_{ijk}-d_{ijk}$ and $C_{ijk}$.

Due to the positive correlation between the number of variables and the amount of transmitted data, the linear program (\ref{coflow:main}) cannot solve the coflow scheduling problem instances with polynomial time complexity in the input size. Therefore, the runtime of algorithms using the optimal linear program solution is only pseudo-polynomial. However, we address this drawback by changing the time intervals to a polynomial number of intervals that geometrically increase in size. We will discuss the details of this technique in Section~\ref{sec:Algorithm3}.

\begin{algorithm}
\caption{Randomized Coflow Scheduling}
    \begin{algorithmic}[1]
				\STATE Compute an optimal solution $y$ to linear programming (\ref{coflow:main}). \label{alg1-2}
				\STATE For all flows $(i, j, k)\in \mathcal{F}$ assign flow $(i, j, k)$ to time $t$, where the time $t$ is chosen from the probability distribution that assigns flow $(i, j, k)$ to $t$ with probability $\frac{y_{ijkt}}{d_{ijk}}$; set $t_{ijk}=t$. \label{alg1-3}
		    \STATE wait until the first coflow is released \label{alg1-4}
				\WHILE{there is some incomplete flow}
            \FOR{every released and incomplete flow $(i, j, k)\in \mathcal{F}$ in non-decreasing order of $t_{ijk}$, breaking ties independently at random}
								\IF{the link $(i, j)$ is idle}
								    \STATE schedule flow $(i, j, k)$\label{alg1-1}
								\ENDIF
						\ENDFOR
						\WHILE{no new flow is completed or released}
						    \STATE transmit the flows that get scheduled in line \ref{alg1-1} at maximum rate 1.
						\ENDWHILE
				\ENDWHILE \label{alg1-5}
   \end{algorithmic}
\label{Alg1}
\end{algorithm}

The randomized coflow scheduling algorithm, presented in Algorithm~\ref{Alg1}, utilizes a list scheduling rule. In lines \ref{alg1-2}-\ref{alg1-3}, we randomly determine the start transmission time for each flow based on the solution to the linear programming (\ref{coflow:main}). It is important to note that we do not constrain the flows to start transmission at this specific time; rather, we use this time to determine the transmission order of the flows. In lines \ref{alg1-4}-\ref{alg1-5}, the flows are scheduled in non-decreasing order of $t_{ijk}$, with ties being broken independently at random. This scheduling ensures that all flows are transmitted in a preemptible manner.

\subsection{Analysis}
In this section, we validate the effectiveness of Algorithm~\ref{Alg1} by establishing its approximation ratios. More precisely, we illustrate that the algorithm attains an expected approximation ratio of $3$ for arbitrary release times and an expected approximation ratio of $2$ in the absence of release times.
In the algorithm analysis, we make the assumption that the random decisions for distinct flows in line \ref{alg1-3} of Algorithm~\ref{Alg1} exhibit pairwise independence.
\begin{lem}\label{lem:lem1}
If the last completed flow in coflow $k$ is $(i, j, k)$, the expected completion time $E[C_{k}]$ of coflow $k$ in the schedule constructed by Algorithm~\ref{Alg1} can be bounded from above by 
\begin{eqnarray*}
E[C_{k}] \leq 3\sum_{t=r_{k}}^{T} \frac{y_{ijkt}}{d_{ijk}}\left(t+\frac{1}{2}\right) + \sum_{t=r_{k}}^{T} y_{ijkt}.
\end{eqnarray*}
If all coflows are released at time 0, the following stronger bound holds:
\begin{eqnarray*}
E[C_{k}] \leq 2\sum_{t=r_{k}}^{T} \frac{y_{ijkt}}{d_{ijk}}\left(t+\frac{1}{2}\right) + \sum_{t=r_{k}}^{T} y_{ijkt}
\end{eqnarray*}
\end{lem}
\begin{proof}
Given that the most recent completed flow in coflow $k$ is denoted as $(i, j, k)$, we have $C_{k}=C_{ijk}$. Consider an arbitrary yet fixed flow $(i, j, k)\in \mathcal{F}$ and let $t$ represent the time to which flow $(i, j, k)$ has been assigned. Let $\tau\geq 0$ be the earliest point in time such that there is no idle time in the constructed schedule during the interval $(\tau, C_{k}]$. Define $\mathcal{P}_{i}$ as the set of flows transmitted on port $i$ during this time interval, and let $\mathcal{P}_{j}$ be the set of flows transmitted on port $j$ during the same interval. We have
\begin{eqnarray}\label{lem1:eq1}
C_{k}-\tau \leq \sum_{(i', j', k')\in \mathcal{P}_{i}\cup \mathcal{P}_{j}} d_{i'j'k'}.
\end{eqnarray}

Since all flows $(i', j', k')\in \mathcal{P}_{i}\cup \mathcal{P}_{j}$ are started no later than flow $(i, j, k)$, their assigned times must satisfy $t_{i'j'k'}\leq t$. Specifically, for all $(i', j', k')\in \mathcal{P}_{i}\cup \mathcal{P}_{j}$, it holds that $r_{k'}\leq t_{i'j'k'}\leq t$. Combining this with (\ref{lem1:eq1}), we deduce inequality
\begin{eqnarray*}\label{lem1:eq2}
C_{k} \leq t+\sum_{(i', j', k')\in \mathcal{P}_{i}\cup \mathcal{P}_{j}} d_{i'j'k'}.
\end{eqnarray*}
In analyzing the expected completion time $E[C_{k}]$ for coflow $k$, we initially hold the assignment of flow $(i, j, k)$ to time $t$ constant and establish an upper bound on the conditional expectation $E_{t}[C_{k}]$:
\begin{eqnarray}\label{lem1:eq3}
E_{t}[C_{k}]  & \leq & t+E_{t}\left[\sum_{(i', j', k')\in \mathcal{P}_{i}\cup \mathcal{P}_{j}} d_{i'j'k'}\right] \notag\\
              & \leq & t+d_{ijk}  \notag\\
					    &      & +\sum_{(i', j', k')\in \mathcal{P}_{i}\setminus \left\{(i, j, k)\right\}} d_{i'j'k'} \cdot Pr_{t}(i', j', k') \notag\\
					    &      & +\sum_{(i', j', k')\in\mathcal{P}_{j}\setminus \left\{(i, j, k)\right\}} d_{i'j'k'} \cdot Pr_{t}(i', j', k') \notag\\
							& \leq & t + d_{ijk} +2\left(t+\frac{1}{2}\right) \notag\\
							& \leq & 3 \left(t+\frac{1}{2}\right)+ d_{ijk}
\end{eqnarray}
where $Pr_{t}(i', j', k')=Pr_{t}[(i', j', k') \mbox{~before~} (i, j, k)]=\sum_{\ell=r_{k'}}^{t-1}\frac{y_{i'j'k'\ell}}{d_{i'j'k'}}+\frac{1}{2}\frac{y_{i'j'k't}}{d_{i'j'k'}}$.
Finally, applying the formula of total expectation to eliminate conditioning results in inequality (\ref{lem1:eq3}). We have
\begin{eqnarray*}\label{lem1:eq4}
E[C_{k}] & =     & \sum_{t=r_{k}}^{T} \frac{y_{ijkt}}{d_{ijk}}E_{t}[C_{k}]  \\
         & \leq  & 3\sum_{t=r_{k}}^{T} \frac{y_{ijkt}}{d_{ijk}}\left(t+\frac{1}{2}\right) + \sum_{t=r_{k}}^{T} y_{ijkt}.
\end{eqnarray*}

In the absence of nontrivial release times, a more robust bound can be established. Note that in this case, $\tau=0$ and $C_{k} \leq \sum_{(i', j', k')\in \mathcal{P}_{i}\cup \mathcal{P}_{j}} d_{i'j'k'}$. This yields
\begin{eqnarray*}\label{lem1:eq5}
E[C_{k}] & \leq  & 2\sum_{t=r_{k}}^{T} \frac{y_{ijkt}}{d_{ijk}}\left(t+\frac{1}{2}\right) + \sum_{t=r_{k}}^{T} y_{ijkt}.
\end{eqnarray*}
and eventually the claimed result.
\end{proof}

\begin{thm}\label{thm:thm1}
For the coflow scheduling problem with release times, the expected objective function value of the schedule constructed by Algorithm~\ref{Alg1} is at most three times the value of an optimal solution.
\end{thm}
\begin{proof}
Lemma~\ref{lem:lem1}, in conjunction with constraint (\ref{coflow:d}), establishes that the expected completion time for every coflow $k$ is bounded by three times its linear programming completion time $C_{k}^*$. Since the optimal solution value of linear programming serves as a lower bound on the optimal schedule value, and considering the nonnegativity of weights, this conclusion follows from the linearity of expectations.
\end{proof}

\begin{thm}\label{thm:thm2}
For the coflow scheduling problem without release times, the expected objective function value of the schedule constructed by Algorithm~\ref{Alg1} is at most twice the value of an optimal solution.
\end{thm}
\begin{proof}
The claimed result follws from Lemma~\ref{lem:lem1} and the linear programming constraint (\ref{coflow:d}).							
\end{proof}

\section{Deterministic Approximation Algorithm for the Coflow Scheduling Problem}\label{sec:Algorithm2}
This section introduces a deterministic algorithm with a bounded worst-case ratio. The randomized algorithms presented in this paper can be derandomized using the method of conditional probabilities, as inspired by Schulz and Skutella~\cite{Schulz2002}. This approach sequentially evaluates random decisions and chooses the most promising one at each decision point. We make the assumption that all uncertain decisions are random. Consequently, if the expected value of the objective function is minimal under the corresponding conditions, this decision is selected.

Algorithm~\ref{Alg2} is derived from Algorithm~\ref{Alg1}. In lines~\ref{alg2-2}-\ref{alg2-3}, we select a time $t$ for each flow to minimize the expected total weighted completion time. In line~\ref{alg2-4}, all flows break ties with smaller indices. We say $ (i', j', k') < (i, j, k)$ if $k' < k$, or $j' < j$ and $k' = k$, or $i' < i$ and $j' = j$ and $k' = k$.

\begin{algorithm}
\caption{Deterministic Coflow Scheduling}
    \begin{algorithmic}[1]
				\STATE Compute an optimal solution $y$ to linear programming (\ref{coflow:main}).
				\STATE Set $\mathcal{P}=\emptyset$; $x=0$; \label{alg2-2}
				\FOR{all $(i, j, k)\in F$} \label{alg2-5}
					\STATE for all possible assignments of $(i, j, k)\in \mathcal{F}\setminus \mathcal{P}$ to times $t$ compute $E_{\mathcal{P}\cup \left\{(i, j, k)\right\}, x}[\sum_{\ell}w_{\ell} C_{\ell}]$; 
					\STATE Determine the time that minimizes the conditional expectation $E_{\mathcal{P}\cup \left\{(i, j, k)\right\}, x}[\sum_{\ell}w_{\ell} C_{\ell}]$
					\STATE Set $\mathcal{P}=\mathcal{P}\cup \left\{(i, j, k)\right\}$; $x_{ijkt}=1$; $t_{ijk}=t$;
				\ENDFOR \label{alg2-3}
		    \STATE wait until the first coflow is released
				\WHILE{there is some incomplete flow}
            \FOR{every released and incomplete flow $(i, j, k)\in \mathcal{F}$ in non-decreasing order of $t_{ijk}$, breaking ties with smaller indices} \label{alg2-4}
								\IF{the link $(i, j)$ is idle}
								    \STATE schedule flow $(i, j, k)$\label{alg2-1}
								\ENDIF
						\ENDFOR
						\WHILE{no new flow is completed or released}
						    \STATE transmit the flows that get scheduled in line \ref{alg2-1} at maximum rate 1.
						\ENDWHILE
				\ENDWHILE
   \end{algorithmic}
\label{Alg2}
\end{algorithm}

Let 
\begin{eqnarray*}
C(i,j,k,t) & = & t+d_{ijk} \\
           &   & +\sum_{(i', j', k')\in \mathcal{F}_{i}\setminus \left\{(i, j, k)\right\}} \sum_{\ell=r_{k'}}^{t-1}y_{i'j'k'\ell} \\
					 &   & +\sum_{(i, j', k')<(i, j, k)} y_{ij'k't} \\
					 &   & +\sum_{(i', j', k')\in\mathcal{F}_{j}\setminus \left\{(i, j, k)\right\}} \sum_{\ell=r_{k'}}^{t-1} y_{i'j'k'\ell} \\
					 &   & +\sum_{(i', j, k')<(i, j, k)} y_{i'jk't}
\end{eqnarray*}
be conditional expectation completion time of flow $(i, j, k)$ to which flow $(i, j, k)$ has been assigned by $t$. The expected completion time of coflow $k$ in the schedule output by Algorithm~\ref{Alg2} is
\begin{eqnarray*}
E[C_{k}] & = & \max_{(i, j, k)\in \mathcal{F}_{k}} \left\{\sum_{t=r_{k}}^{T} \frac{y_{ijkt}}{d_{ijk}} C(i,j,k,t)\right\}.
\end{eqnarray*}

Let $\mathcal{P}_{i}\subseteq F_{i}$ and $\mathcal{P}_{j}\subseteq F_{j}$ represent subsets of flows that have already been assigned the time. 
For each flow $(i', j', k')\in \mathcal{P}_{i} \cup \mathcal{P}_{j}$, let the 0/1-variable $x_{i'j'k't}$ for $t\geq r_{k'}$ indicate whether $(i', j', k')$ has been assigned to the time $t$ (i.e., $x_{i'j'k't}=1$) or not ($x_{i'j'k't}=0$). This allows us to formulate the following expressions for the conditional expectation of the completion time of $(i, j, k)$. 
Let
\begin{eqnarray*}
D(i,j,k,t) & = & t+d_{ijk} \\
           &   & +\sum_{(i', j', k')\in \mathcal{P}_{i}\cup  \mathcal{P}_{j}} \sum_{\ell=r_{k'}}^{t-1}x_{i'j'k'\ell}d_{i'j'k'} \\
					 &   & +\sum_{\substack{(i', j', k')\in \mathcal{P}_{i}\cup  \mathcal{P}_{j}\\ (i', j', k')<(i, j, k)}} x_{i'j'k't}d_{i'j'k'} \\
					 &   & +\sum_{\substack{(i', j', k')\in \\ (\mathcal{F}_{i}\cup \mathcal{F}_{j})\setminus \left\{\mathcal{P}_{i}\cup \mathcal{P}_{j}\cup (i, j, k)\right\}}} \sum_{\ell=r_{k'}}^{t-1} y_{i'j'k'\ell} \\
					 &   & +\sum_{\substack{(i', j', k')\in \\ (\mathcal{F}_{i}\cup \mathcal{F}_{j})\setminus \left\{\mathcal{P}_{i}\cup \mathcal{P}_{j}\cup (i, j, k)\right\}\\ (i', j', k')<(i, j, k)}} y_{i'j'k't}
\end{eqnarray*}
be conditional expectation completion time of flow $(i, j, k)$ to which flow $(i, j, k)$ has been assigned by $t$ when $(i, j, k)\notin \mathcal{P}_{i}\cup \mathcal{P}_{j}$ and let
\begin{eqnarray*}
E(i,j,k,t) & = & t+d_{ijk} \\
           &   & +\sum_{(i', j', k')\in \mathcal{P}_{i}\cup  \mathcal{P}_{j}} \sum_{\ell=r_{k'}}^{t-1}x_{i'j'k'\ell}d_{i'j'k'} \\
					 &   & +\sum_{\substack{(i', j', k')\in \mathcal{P}_{i}\cup  \mathcal{P}_{j}\\ (i', j', k')<(i, j, k)}} x_{i'j'k't}d_{i'j'k'} \\
					 &   & +\sum_{\substack{(i', j', k')\in \\(\mathcal{F}_{i}\cup \mathcal{F}_{j})\setminus \left\{\mathcal{P}_{i}\cup \mathcal{P}_{j}\right\}}} \sum_{\ell=r_{k'}}^{t-1} y_{i'j'k'\ell} \\
					 &   & +\sum_{\substack{(i', j', k')\in \\(\mathcal{F}_{i}\cup \mathcal{F}_{j})\setminus \left\{\mathcal{P}_{i}\cup \mathcal{P}_{j}\right\} \\ (i', j', k')<(i, j, k)}} y_{i'j'k't}
\end{eqnarray*}
be expectation completion time of flow $(i, j, k)$ to which flow $(i, j, k)$ has been assigned by $t$ when $(i, j, k)\in \mathcal{P}_{i}\cup \mathcal{P}_{j}$.
Let $\mathcal{P}\subseteq F$ represent subsets of flows that have already been assigned the switch-interval pair.
The expected completion time of coflow $k$ is the maximum expected completion time among its flows. We have
\begin{eqnarray}\label{eq1}
E_{\mathcal{P},x}[C_{k}] & = & \max \left\{A, B\right\}
\end{eqnarray}
where
\begin{eqnarray*}
A & = & \max_{(i, j, k)\in \mathcal{F}_{k}\setminus \mathcal{P}} \left\{\sum_{t=r_{k}}^{T} \frac{y_{ijkt}}{d_{ijk}} D(i,j,k,t)\right\},
\end{eqnarray*}
\begin{eqnarray}\label{B}
B & = & \max_{(i, j, k)\in \mathcal{F}_{k}\cap \mathcal{P}} \left\{E(i,j,k,t_{ijk}) \right\}.
\end{eqnarray}
In equation~(\ref{B}), $t_{ijk}$ is the $(i, j, k)$ has been assigned to, i.e., $x_{ijkt_{ijk}}=1$.

\begin{lem}\label{lem:lem2}
Let $y$ be an optimal solution to linear programming (\ref{coflow:main}), $\mathcal{P}\subseteq \mathcal{F}$, and let $x$ represent a fixed assignment of the flows in $\mathcal{P}$ to times. Moreover, for $(i, j, k)\in \mathcal{F}\setminus \mathcal{P}$, there exists an assignment of $(i, j, k)$ to time $t$ with $r_{k}\leq t$ such that
\begin{eqnarray*}
E_{\mathcal{P}\cup \left\{(i, j, k)\right\}, x}\left[\sum_{\ell}w_{\ell} C_{\ell}\right] \leq E_{\mathcal{P}, x}\left[\sum_{\ell}w_{\ell} C_{\ell}\right].
\end{eqnarray*}
\end{lem}
\begin{proof}
The expression for the conditional expectation, $E_{\mathcal{P}, x}[\sum_{\ell}w_{\ell} C_{\ell}]$, can be expressed as a convex combination of conditional expectations $E_{\mathcal{P}\cup \left\{(i, j, k)\right\}, x}[\sum_{\ell}w_{\ell} C_{\ell}]$ across all possible assignments of flow $(i, j, k)$ to time $t$, where the coefficients are given by $\frac{y_{ijkt}}{d_{ijk}}$. The optimal combination is determined by the condition $E_{\mathcal{P}\cup \left\{(i, j, k)\right\}, x}[\sum_{\ell}w_{\ell} C_{\ell}] \leq E_{\mathcal{P}, x}[\sum_{\ell}w_{\ell} C_{\ell}]$, eventually the claimed result.			
\end{proof}

\begin{thm}\label{thm:thm3}
Algorithm~\ref{Alg2} is a deterministic algorithm with performance guarantee 3 for the coflow scheduling problem with release times and with performance guarantee $2+\frac{1}{LB}$ for the coflow scheduling problem without release times where
\begin{eqnarray*}
LB=\min_k\max\left\{\max_{i}\left\{L_{ik}\right\},\max_{j}\left\{L_{jk}\right\}\right\}.
\end{eqnarray*}
Moreover, the running time of this algorithm is polynomial in the number of variables of linear programming (\ref{coflow:main}).
\end{thm}
\begin{proof}
We consider an arbitrary but fixed coflow $k\in \mathcal{K}$ and denote the time $t$ to which flow $(i, j, k)$ has been assigned by $t$.
When breaking ties with smaller indices, we have 
\begin{eqnarray*}\label{thm3:eq1}
E_{t}[C_{k}]  & \leq & t + d_{ijk} +2\left(t+1\right) \\
							& \leq & 3 \left(t+\frac{1}{2}\right)+ d_{ijk}+\frac{1}{2}
\end{eqnarray*}
and 
\begin{eqnarray*}\label{thm3:eq2}
E[C_{k}] & =     & \sum_{t=r_{k}}^{T} \frac{y_{ijkt}}{d_{ijk}}E_{t}[C_{k}]  \\
         & \leq  & 3\sum_{t=r_{k}}^{T} \left(\frac{y_{ijkt}}{d_{ijk}}\left(t+\frac{1}{2}\right) + \frac{1}{2} y_{ijkt}\right) \\
				 & \leq  & 3 C_{k}^*
\end{eqnarray*}
for the coflow scheduling problem with release times. We also have 
\begin{eqnarray*}\label{thm3:eq3}
E_{t}[C_{k}]  & \leq & d_{ijk} + 2\left(t+1\right) \\
							& \leq & 2 \left(t+\frac{1}{2}\right)+ d_{ijk}+1
\end{eqnarray*}
and 
\begin{eqnarray*}\label{thm3:eq4}
E[C_{k}] & =     & \sum_{t=r_{k}}^{T} \frac{y_{ijkt}}{d_{ijk}}E_{t}[C_{k}]  \\
         & \leq  & 2\sum_{t=r_{k}}^{T} \left(\frac{y_{ijkt}}{d_{ijk}}\left(t+\frac{1}{2}\right) + \frac{1}{2} y_{ijkt}\right)+1 \\
				 & \leq  & 2 C_{k}^*+1 \\
				 & \leq  & \left(2+\frac{1}{LB_{k}}\right)C_{k}^*
\end{eqnarray*}
for the coflow scheduling problem without release times where $LB_{k}=\max\left\{\max_{i}\left\{L_{ik}\right\},\max_{j}\left\{L_{jk}\right\}\right\}$ is a lower bound of $C_{k}^*$.
Subsequently, the result is derived through the inductive application of Lemma~\ref{lem:lem2}. 

The calculation of equation (\ref{eq1}) is polynomially bounded by the number of variables in linear programming (\ref{coflow:main}). Consequently, the running time for each iteration of the number of flows in lines~\ref{alg2-5}-\ref{alg2-3} of Algorithm~\ref{Alg1} is also polynomially bounded by this number.
\end{proof}

\section{Interval-indexed Linear Programming Relaxations}\label{sec:Algorithm3}
This section introduces interval-indexed linear programming to address the issue of excessive variable quantities in linear programming (\ref{coflow:main}). We initially propose a randomized approximation algorithm. This approach can also be modified into a deterministic approximation algorithm, as discussed in section~\ref{sec:Algorithm2}, while still achieving the same approximation ratios.

For a given positive parameter $\eta$, we select the integer $L$ as the smallest value satisfying the inequality $(1+\eta)^L\geq T+1$. As a result, the value of $L$ is polynomially bounded in the input size of the scheduling problem under consideration. Let $\mathcal{L}=\left\{0, 1, \ldots, L\right\}$ be the set of time interval indices. Consider the interval $I_{0}=[0, 1]$, and for $1\leq \ell \leq L$, define $I_{\ell}=((1+\eta)^{\ell-1},(1+\eta)^{\ell}]$. The length of the $\ell$th interval, denoted as $|I_{\ell}|$, is given by $|I_{\ell}|=\eta(1+\eta)^{\ell-1}$ for $1\leq \ell \leq L$. We introduce variables $y_{ijk\ell}$ for $(i, j, k)\in \mathcal{F}$, subject to the constraint $(1+\eta)^{\ell-1}\geq r_{k}$. The expression $y_{ijk\ell}\cdot |I_{\ell}|$ represents the time during which flow $(i, j, k)$ is transmitted within the time interval $I_{\ell}$. Equivalently, $(y_{ijk\ell}\cdot |I_{\ell}|)/d_{ijk}$ denotes the fraction of flow $(i, j, k)$ transmitted within $I_{\ell}$. Consider the following linear program in these interval-indexed variables:

\begin{subequations}\label{coflow:interval}
\begin{align}
& \text{min}  && \sum_{k \in \mathcal{K}} w_{k} C_{k}     &   & \tag{\ref{coflow:interval}} \\
& \text{s.t.} && \sum_{\substack{\ell=0\\ (1+\eta)^{\ell-1}\geq r_{k}}}^{L} \frac{y_{ijk\ell}\cdot |I_{\ell}|}{d_{ijk}} = 1, && \forall (i, j, k)\in \mathcal{F} \label{interval:a} \\
&  && \sum_{(i, j, k) \in \mathcal{F}_{i}} y_{ijk\ell} \leq 1, && \forall i\in \mathcal{I}, \forall \ell\in \mathcal{L} \label{interval:b} \\
&  && \sum_{(i, j, k) \in \mathcal{F}_{j}} y_{ijk\ell} \leq 1, && \forall j\in \mathcal{J}, \forall \ell\in \mathcal{L} \label{interval:c} \\
&  && C_{k}\geq C_{ijk}, && \forall k\in \mathcal{K}, \forall (i, j, k)\in \mathcal{F}_{k} \label{interval:d} \\
&  && y_{ijk\ell} \geq 0,&& \forall (i, j, k)\in \mathcal{F}, \notag\\
&  &&                    && \forall \ell\in \mathcal{L}\label{interval:g}
\end{align}
\end{subequations}
where
\begin{eqnarray}\label{interval:d2}
C_{ijk} = \sum_{\substack{\ell=0\\ (1+\eta)^{\ell-1}\geq r_{k}}}^{L}\left(\frac{y_{ijk\ell}\cdot |I_{\ell}|}{d_{ijk}}(1+\eta)^{\ell-1}+\frac{1}{2}y_{ijk\ell}\cdot |I_{\ell}|\right).
\end{eqnarray}
To simplify the notation of equation (\ref{interval:d2}), we define the notation $(1+\eta)^{\ell-1}$ to be $1/2$ for $\ell=0$ in the term $\frac{y_{ijk\ell}\cdot |I_{\ell}|}{d_{ijk}}(1+\eta)^{\ell-1}$ of the equation. Constraint (\ref{interval:a}) ensures the fulfillment of the transmitting requirements for each flow. The switch capacity constraints (\ref{interval:b}) and (\ref{interval:c}) stipulate that the switch can transmit only one flow at a time in each input port $i$ and output port $j$. Constraint (\ref{interval:d}) specifies that the completion time of coflow $k$ is constrained by the completion times of all its constituent flows. Additionally, the equation (\ref{interval:d2}) represents the lower bound of the completion time for flow $(i, j, k)$, and this lower bound occurs during continuous transmission between $C_{ijk}-d_{ijk}$ and $C_{ijk}$.

\begin{algorithm}
\caption{Randomized Interval-Indexed Coflow Scheduling}
    \begin{algorithmic}[1]
				\STATE Compute an optimal solution $y$ to linear programming (\ref{coflow:interval}). \label{alg3-2}
				\STATE For all flows $(i, j, k)\in \mathcal{F}$ assign flow $(i, j, k)$ to interval $I_{\ell}$, where the interval $I_{\ell}$ is chosen from the probability distribution that assigns flow $(i, j, k)$ to $I_{\ell}$ with probability $\frac{y_{ijk\ell}\cdot |I_{\ell}|}{d_{ijk}}$; set $t_{ijk}$ to the left endpoint of time interval $I_{\ell}$. \label{alg3-3}
		    \STATE wait until the first coflow is released \label{alg3-4}
				\WHILE{there is some incomplete flow}
            \FOR{every released and incomplete flow $(i, j, k)\in \mathcal{F}$ in non-decreasing order of $t_{ijk}$, breaking ties independently at random}
								\IF{the link $(i, j)$ is idle}
								    \STATE schedule flow $(i, j, k)$\label{alg3-1}
								\ENDIF
						\ENDFOR
						\WHILE{no new flow is completed or released}
						    \STATE transmit the flows that get scheduled in line \ref{alg3-1} at maximum rate 1.
						\ENDWHILE
				\ENDWHILE \label{alg3-5}
   \end{algorithmic}
\label{Alg3}
\end{algorithm}

The randomized interval-indexed coflow scheduling algorithm, presented in Algorithm~\ref{Alg3}, utilizes a list scheduling rule. In lines \ref{alg3-2}-\ref{alg3-3}, we randomly determine the start transmission time for each flow based on the solution to the linear programming (\ref{coflow:interval}). In lines \ref{alg3-4}-\ref{alg3-5}, the flows are scheduled in non-decreasing order of $t_{ijk}$, with ties being broken independently at random.

\begin{thm}\label{thm:thm4}
The expected completion time of each coflow $k$ in the schedule constructed by Algorithm~\ref{Alg3} is at most $3(1+\frac{\eta}{3})C_{k}^{*}$ for the coflow scheduling problem with release times and is at most $2(1+\frac{\eta}{2})C_{k}^{*}$ for the coflow scheduling problem without release times.
\end{thm}
\begin{proof}
With similar arguments as in the proof of Lemma~\ref{lem:lem1}, the assignment of flow $(i, j, k)$ to time interval $I_{0}$ constant and establish an upper bound on the conditional expectation $E_{\ell=0}[C_{k}]$:
\begin{eqnarray}\label{thm4:eq3}
E_{\ell=0}[C_{k}]  & \leq & E_{\ell=0}\left[\sum_{(i', j', k')\in \mathcal{P}_{i}\cup \mathcal{P}_{j}} d_{i'j'k'}\right] \notag\\
              & \leq & d_{ijk} +\sum_{(i', j', k')\in \mathcal{P}_{i}\setminus \left\{(i, j, k)\right\}} \frac{1}{2}y_{i'j'k'0}\notag\\
					    &      & +\sum_{(i', j', k')\in\mathcal{P}_{j}\setminus \left\{(i, j, k)\right\}}\frac{1}{2}y_{i'j'k'0} \notag\\
							& \leq & d_{ijk} + 1 \notag\\
							& =    & 3\left(1+\frac{\eta}{3}\right)\left(\frac{d_{ijk} + 1}{3(1+\frac{\eta}{3})}\right) \notag\\
							& \leq & 3\left(1+\frac{\eta}{3}\right)\left(\frac{1}{2}+\frac{1}{2}d_{ijk}\right) \notag\\
							& =    & 3\left(1+\frac{\eta}{3}\right)\left((1+\eta)^{\ell-1}+\frac{1}{2}d_{ijk}\right).
\end{eqnarray}
Recall that the notation $(1+\eta)^{\ell-1}$ is set to $1/2$ for $\ell=0$.
The assignment of flow $(i, j, k)$ to time interval $I_{\ell}$ constant and establish an upper bound on the conditional expectation $E_{\ell}[C_{k}]$:
\begin{eqnarray}\label{thm4:eq4}
E_{\ell}[C_{k}]  & \leq & (1+\eta)^{\ell-1}+E_{\ell}\left[\sum_{(i', j', k')\in \mathcal{P}_{i}\cup \mathcal{P}_{j}} d_{i'j'k'}\right] \notag\\
              & \leq & (1+\eta)^{\ell-1}+d_{ijk}  \notag\\
					    &      & +\sum_{(i', j', k')\in \mathcal{P}_{i}\setminus \left\{(i, j, k)\right\}} d_{i'j'k'} \cdot Pr_{\ell}(i', j', k') \notag\\
					    &      & +\sum_{(i', j', k')\in\mathcal{P}_{j}\setminus \left\{(i, j, k)\right\}} d_{i'j'k'} \cdot Pr_{\ell}(i', j', k') \notag\\
              & \leq & (1+\eta)^{\ell-1}+d_{ijk} +2 \sum_{t=r_{k'}}^{\ell-1}|I_{t}|+|I_{\ell}| \notag\\
							& \leq & 3\left(1+\frac{\eta}{3}\right)(1+\eta)^{\ell-1}+d_{ijk} \notag\\
							& \leq & 3\left(1+\frac{\eta}{3}\right)\left((1+\eta)^{\ell-1}+\frac{1}{2}d_{ijk}\right)
\end{eqnarray}
where $Pr_{\ell}(i', j', k')=\sum_{t=r_{k'}}^{\ell-1}\frac{y_{i'j'k't}\cdot |I_{t}|}{d_{i'j'k'}}+\frac{1}{2}\frac{y_{i'j'k'\ell}\cdot |I_{\ell}|}{d_{i'j'k'}}$.

Finally, applying the formula of total expectation to eliminate conditioning results in inequalities (\ref{thm4:eq3}) and (\ref{thm4:eq4}). We have
\begin{eqnarray*}\label{thm4:eq1}
 E[C_{k}]  & \leq & 3 \left(1+\frac{\eta}{3}\right)   \sum_{\substack{\ell=0\\ (1+\eta)^{\ell-1}\geq r_{k}}}^{L}\left(\frac{y_{ijk\ell}\cdot |I_{\ell}|}{d_{ijk}}(1+\eta)^{\ell-1}+\frac{1}{2}y_{ijk\ell}\cdot |I_{\ell}|\right)
\end{eqnarray*}
for the coflow scheduling problem with release times. We also can obtain
\begin{eqnarray*}\label{thm4:eq2}
E[C_{k}] & \leq  & 2 \left(1+\frac{\eta}{2}\right)  \sum_{\substack{\ell=0\\ (1+\eta)^{\ell-1}\geq r_{k}}}^{L}\left(\frac{y_{ijk\ell}\cdot |I_{\ell}|}{d_{ijk}}(1+\eta)^{\ell-1}+\frac{1}{2}y_{ijk\ell}\cdot |I_{\ell}|\right)
\end{eqnarray*}
for the coflow scheduling problem without release times. This together with constraints (\ref{interval:d}) yields the theorem.
\end{proof}

We can choose $\eta=\epsilon$ for any given positive parameter $\epsilon>0$. Algorithm~\ref{Alg3} achieves expected approximation ratios of $3+\epsilon$ and $2+\epsilon$ in the cases of arbitrary and zero release times, respectively.

\section{Deterministic Approximation Algorithm with Interval-indexed Approach for the Coflow Scheduling Problem}\label{sec:Algorithm4}
This section introduces a deterministic algorithm with a bounded worst-case ratio. Algorithm~\ref{Alg4} is derived from Algorithm~\ref{Alg3}. In lines~\ref{alg4-2}-\ref{alg4-3}, we select a time interval $\ell$ for each flow to minimize the expected total weighted completion time. In line~\ref{alg4-4}, all flows break ties with smaller indices. 

\begin{algorithm}
\caption{Deterministic Interval-Indexed Coflow Scheduling}
    \begin{algorithmic}[1]
				\STATE Compute an optimal solution $y$ to linear programming (\ref{coflow:interval}).
				\STATE Set $\mathcal{P}=\emptyset$; $x=0$; \label{alg4-2}
				\FOR{all $(i, j, k)\in F$} \label{alg4-5}
					\STATE for all possible assignments of $(i, j, k)\in \mathcal{F}\setminus \mathcal{P}$ to times $\ell$ compute $E_{\mathcal{P}\cup \left\{(i, j, k)\right\}, x}[\sum_{q}w_{q} C_{q}]$; 
					\STATE Determine the time that minimizes the conditional expectation $E_{\mathcal{P}\cup \left\{(i, j, k)\right\}, x}[\sum_{q}w_{q} C_{q}]$
					\STATE Set $\mathcal{P}=\mathcal{P}\cup \left\{(i, j, k)\right\}$; $x_{ijk\ell}=1$; set $t_{ijk}$ to the left endpoint of time interval $I_{\ell}$;
				\ENDFOR \label{alg4-3}
		    \STATE wait until the first coflow is released
				\WHILE{there is some incomplete flow}
            \FOR{every released and incomplete flow $(i, j, k)\in \mathcal{F}$ in non-decreasing order of $t_{ijk}$, breaking ties with smaller indices} \label{alg4-4}
								\IF{the link $(i, j)$ is idle}
								    \STATE schedule flow $(i, j, k)$\label{alg4-1}
								\ENDIF
						\ENDFOR
						\WHILE{no new flow is completed or released}
						    \STATE transmit the flows that get scheduled in line \ref{alg4-1} at maximum rate 1.
						\ENDWHILE
				\ENDWHILE
   \end{algorithmic}
\label{Alg4}
\end{algorithm}

Let 
\begin{eqnarray*}
C(i,j,k,0) & = & d_{ijk} +\sum_{(i, j', k')<(i, j, k)} y_{ij'k'0} \\
           &   & +\sum_{(i', j, k')<(i, j, k)} y_{i'jk'0}
\end{eqnarray*}
be conditional expectation completion time of flow $(i, j, k)$ to which flow $(i, j, k)$ has been assigned by $\ell=0$ and let 
\begin{eqnarray*}
C(i,j,k,\ell) & = & (1+\eta)^{\ell-1}+d_{ijk} \\
           &   & +\sum_{(i', j', k')\in \mathcal{F}_{i}\setminus \left\{(i, j, k)\right\}} \sum_{\substack{t=0\\ (1+\eta)^{t-1}\geq r_{k'}}}^{\ell-1}y_{i'j'k't}\cdot |I_{t}| \\
					 &   & +\sum_{(i, j', k')<(i, j, k)} y_{ij'k'\ell}\cdot |I_{\ell}| \\
					 &   & +\sum_{(i', j', k')\in\mathcal{F}_{j}\setminus \left\{(i, j, k)\right\}} \sum_{\substack{t=0\\ (1+\eta)^{t-1}\geq r_{k'}}}^{\ell-1} y_{i'j'k't}\cdot |I_{t}| \\
					 &   & +\sum_{(i', j, k')<(i, j, k)} y_{i'jk'\ell}\cdot |I_{\ell}|
\end{eqnarray*}
be conditional expectation completion time of flow $(i, j, k)$ to which flow $(i, j, k)$ has been assigned by $\ell>0$. The expected completion time of coflow $k$ in the schedule output by Algorithm~\ref{Alg4} is
\begin{eqnarray*}
E[C_{k}] & = & \max_{(i, j, k)\in \mathcal{F}_{k}} \left\{\sum_{\substack{\ell=0\\ (1+\eta)^{\ell-1}\geq r_{k}}}^{L} \frac{y_{ijk\ell}\cdot |I_{\ell}|}{d_{ijk}} C(i,j,k,\ell)\right\}.
\end{eqnarray*}

Let $\mathcal{P}_{i}\subseteq F_{i}$ and $\mathcal{P}_{j}\subseteq F_{j}$ represent subsets of flows that have already been assigned the time interval. 
For each flow $(i', j', k')\in \mathcal{P}_{i} \cup \mathcal{P}_{j}$, let the 0/1-variable $x_{i'j'k'\ell}$ for $(1+\eta)^{\ell-1} \geq r_{k'}$ indicate whether $(i', j', k')$ has been assigned to the time interval $I_{\ell}$ (i.e., $x_{i'j'k'\ell}=1$) or not ($x_{i'j'k'\ell}=0$). This allows us to formulate the following expressions for the conditional expectation of the completion time of $(i, j, k)$. 
Let
\begin{eqnarray*}
D(i,j,k,0) & = & d_{ijk}+\sum_{\substack{(i', j', k')\in \mathcal{P}_{i}\cup  \mathcal{P}_{j}\\ (i', j', k')<(i, j, k)}} x_{i'j'k'0}d_{i'j'k'} \\
					 &   & +\sum_{\substack{(i', j', k')\in \\ (\mathcal{F}_{i}\cup \mathcal{F}_{j})\setminus \left\{\mathcal{P}_{i}\cup \mathcal{P}_{j}\cup (i, j, k)\right\}\\ (i', j', k')<(i, j, k)}} y_{i'j'k'0}
\end{eqnarray*}
be conditional expectation completion time of flow $(i, j, k)$ to which flow $(i, j, k)$ has been assigned by $\ell=0$ when $(i, j, k)\notin \mathcal{P}_{i}\cup \mathcal{P}_{j}$ and let
\begin{eqnarray*}
D(i,j,k,\ell) & = & (1+\eta)^{\ell-1}+d_{ijk} \\
           &   & +\sum_{(i', j', k')\in \mathcal{P}_{i}\cup  \mathcal{P}_{j}} \sum_{\substack{t=0\\ (1+\eta)^{t-1}\geq r_{k'}}}^{\ell-1}x_{i'j'k't}d_{i'j'k'} \\
					 &   & +\sum_{\substack{(i', j', k')\in \mathcal{P}_{i}\cup  \mathcal{P}_{j}\\ (i', j', k')<(i, j, k)}} x_{i'j'k'\ell}d_{i'j'k'} \\
					 &   & +\sum_{\substack{(i', j', k')\in \\ (\mathcal{F}_{i}\cup \mathcal{F}_{j})\setminus \left\{\mathcal{P}_{i}\cup \mathcal{P}_{j}\cup (i, j, k)\right\}}} \sum_{\substack{t=0\\ (1+\eta)^{t-1}\geq r_{k'}}}^{\ell-1} y_{i'j'k't} \cdot |I_{t}|\\
					 &   & +\sum_{\substack{(i', j', k')\in \\ (\mathcal{F}_{i}\cup \mathcal{F}_{j})\setminus \left\{\mathcal{P}_{i}\cup \mathcal{P}_{j}\cup (i, j, k)\right\}\\ (i', j', k')<(i, j, k)}} y_{i'j'k'\ell}\cdot |I_{\ell}|
\end{eqnarray*}
be conditional expectation completion time of flow $(i, j, k)$ to which flow $(i, j, k)$ has been assigned by $\ell>0$ when $(i, j, k)\notin \mathcal{P}_{i}\cup \mathcal{P}_{j}$.

Let
\begin{eqnarray*}
E(i,j,k,0) & = & d_{ijk}+\sum_{\substack{(i', j', k')\in \mathcal{P}_{i}\cup  \mathcal{P}_{j}\\ (i', j', k')<(i, j, k)}} x_{i'j'k'0}d_{i'j'k'} \\
					 &   & +\sum_{\substack{(i', j', k')\in \\(\mathcal{F}_{i}\cup \mathcal{F}_{j})\setminus \left\{\mathcal{P}_{i}\cup \mathcal{P}_{j}\right\} \\ (i', j', k')<(i, j, k)}} y_{i'j'k'0}\cdot |I_{0}|
\end{eqnarray*}
be expectation completion time of flow $(i, j, k)$ to which flow $(i, j, k)$ has been assigned by $\ell=0$ when $(i, j, k)\in \mathcal{P}_{i}\cup \mathcal{P}_{j}$ and let
\begin{eqnarray*}
E(i,j,k,\ell) & = & (1+\eta)^{\ell-1}+d_{ijk} \\
           &   & +\sum_{(i', j', k')\in \mathcal{P}_{i}\cup  \mathcal{P}_{j}} \sum_{\substack{t=0\\ (1+\eta)^{t-1}\geq r_{k'}}}^{\ell-1}x_{i'j'k't}d_{i'j'k'} \\
					 &   & +\sum_{\substack{(i', j', k')\in \mathcal{P}_{i}\cup  \mathcal{P}_{j}\\ (i', j', k')<(i, j, k)}} x_{i'j'k'\ell}d_{i'j'k'} \\
					 &   & +\sum_{\substack{(i', j', k')\in \\(\mathcal{F}_{i}\cup \mathcal{F}_{j})\setminus \left\{\mathcal{P}_{i}\cup \mathcal{P}_{j}\right\}}} \sum_{\substack{t=0\\ (1+\eta)^{t-1}\geq r_{k'}}}^{\ell-1} y_{i'j'k't} \cdot |I_{t}|\\
					 &   & +\sum_{\substack{(i', j', k')\in \\(\mathcal{F}_{i}\cup \mathcal{F}_{j})\setminus \left\{\mathcal{P}_{i}\cup \mathcal{P}_{j}\right\} \\ (i', j', k')<(i, j, k)}} y_{i'j'k'\ell}\cdot |I_{\ell}|
\end{eqnarray*}
be expectation completion time of flow $(i, j, k)$ to which flow $(i, j, k)$ has been assigned by $\ell>0$ when $(i, j, k)\in \mathcal{P}_{i}\cup \mathcal{P}_{j}$.

Let $\mathcal{P}\subseteq F$ represent subsets of flows that have already been assigned the switch-interval pair.
The expected completion time of coflow $k$ is the maximum expected completion time among its flows. We have
\begin{eqnarray*}
E_{\mathcal{P},x}[C_{k}] & = & \max \left\{A, B\right\}
\end{eqnarray*}
where
\begin{eqnarray*}
A & = & \max_{(i, j, k)\in \mathcal{F}_{k}\setminus \mathcal{P}} \left\{\sum_{\substack{\ell=0\\ (1+\eta)^{\ell-1}\geq r_{k}}}^{L} \frac{y_{ijk\ell}\cdot |I_{\ell}|}{d_{ijk}} D(i,j,k,\ell)\right\},
\end{eqnarray*}
\begin{eqnarray}\label{B2}
B & = & \max_{(i, j, k)\in \mathcal{F}_{k}\cap \mathcal{P}} \left\{E(i,j,k,\ell_{ijk}) \right\}.
\end{eqnarray}
In equation~(\ref{B2}), $\ell_{ijk}$ is the $(i, j, k)$ has been assigned to, i.e., $x_{ijk\ell_{ijk}}=1$.

\begin{lem}\label{lem:lem3}
Let $y$ be an optimal solution to linear programming (\ref{coflow:interval}), $\mathcal{P}\subseteq \mathcal{F}$, and let $x$ represent a fixed assignment of the flows in $\mathcal{P}$ to times. Moreover, for $(i, j, k)\in \mathcal{F}\setminus \mathcal{P}$, there exists an assignment of $(i, j, k)$ to time interval $I_{\ell}$ with $r_{k}\leq (1+\eta)^{\ell-1}$ such that
\begin{eqnarray*}
E_{\mathcal{P}\cup \left\{(i, j, k)\right\}, x}\left[\sum_{q}w_{q} C_{q}\right] \leq E_{\mathcal{P}, x}\left[\sum_{q}w_{q} C_{q}\right].
\end{eqnarray*}
\end{lem}
\begin{proof}
The expression for the conditional expectation, $E_{\mathcal{P}, x}[\sum_{q}w_{q} C_{q}]$, can be expressed as a convex combination of conditional expectations $E_{\mathcal{P}\cup \left\{(i, j, k)\right\}, x}[\sum_{q}w_{q} C_{q}]$ across all possible assignments of flow $(i, j, k)$ to time interval $I_{\ell}$, where the coefficients are given by $\frac{y_{ijk\ell}\cdot |I_{\ell}|}{d_{ijk}}$. The optimal combination is determined by the condition $E_{\mathcal{P}\cup \left\{(i, j, k)\right\}, x}[\sum_{q}w_{q} C_{q}] \leq E_{\mathcal{P}, x}[\sum_{q}w_{q} C_{q}]$, eventually the claimed result.			
\end{proof}

\begin{thm}\label{thm:thm5}
Algorithm~\ref{Alg4} is a deterministic algorithm with performance guarantee $3+\epsilon$ for the coflow scheduling problem with release times and with performance guarantee $2+\epsilon$ for the coflow scheduling problem without release times.
\end{thm}
\begin{proof}

With similar arguments as in the proof of theorem~\ref{thm:thm4}, the assignment of flow $(i, j, k)$ to time interval $I_{0}$ constant and establish an upper bound on the conditional expectation $E_{\ell=0}[C_{k}]$:
\begin{eqnarray}\label{thm5:eq3}
E_{\ell=0}[C_{k}]  & \leq & E_{\ell=0}\left[\sum_{(i', j', k')\in \mathcal{P}_{i}\cup \mathcal{P}_{j}} d_{i'j'k'}\right] \notag\\
              & \leq & d_{ijk} +\sum_{(i', j', k')\in \mathcal{P}_{i}\setminus \left\{(i, j, k)\right\}} y_{i'j'k'0}\notag\\
					    &      & +\sum_{(i', j', k')\in\mathcal{P}_{j}\setminus \left\{(i, j, k)\right\}}y_{i'j'k'0} \notag\\
							& \leq & d_{ijk} + 2 \notag\\
							& =    & 3\left(1+\frac{2\eta}{3}\right)\left(\frac{d_{ijk} + 2}{3(1+\frac{2\eta}{3})}\right) \notag\\
							& \leq & 3\left(1+\frac{2\eta}{3}\right)\left(\frac{1}{2}+\frac{1}{2}d_{ijk}\right) \notag\\
							& =    & 3\left(1+\frac{2\eta}{3}\right)\left((1+\eta)^{\ell-1}+\frac{1}{2}d_{ijk}\right).
\end{eqnarray}
The assignment of flow $(i, j, k)$ to time interval $I_{\ell}$ constant and establish an upper bound on the conditional expectation $E_{\ell}[C_{k}]$:
\begin{eqnarray}\label{thm5:eq4}
E_{\ell}[C_{k}]  & \leq & (1+\eta)^{\ell-1}+E_{\ell}\left[\sum_{(i', j', k')\in \mathcal{P}_{i}\cup \mathcal{P}_{j}} d_{i'j'k'}\right] \notag\\
              & \leq & (1+\eta)^{\ell-1}+d_{ijk}  \notag\\
					    &      & +\sum_{(i', j', k')\in \mathcal{P}_{i}\setminus \left\{(i, j, k)\right\}} d_{i'j'k'} \cdot Pr_{\ell}(i', j', k') \notag\\
					    &      & +\sum_{(i', j', k')\in\mathcal{P}_{j}\setminus \left\{(i, j, k)\right\}} d_{i'j'k'} \cdot Pr_{\ell}(i', j', k') \notag\\
              & \leq & (1+\eta)^{\ell-1}+d_{ijk} +2 \sum_{t=r_{k'}}^{\ell}|I_{t}| \notag\\
							& \leq & 3\left(1+\frac{2\eta}{3}\right)(1+\eta)^{\ell-1}+d_{ijk} \notag\\
							& \leq & 3\left(1+\frac{2\eta}{3}\right)\left((1+\eta)^{\ell-1}+\frac{1}{2}d_{ijk}\right)
\end{eqnarray}
where $Pr_{\ell}(i', j', k')=\sum_{t=r_{k'}}^{\ell}\frac{y_{i'j'k't}\cdot |I_{t}|}{d_{i'j'k'}}$.

Finally, applying the formula of total expectation to eliminate conditioning results in inequalities (\ref{thm5:eq3}) and (\ref{thm5:eq4}). We have
\begin{eqnarray*}\label{thm5:eq1}
 E[C_{k}]  & \leq & 3 \left(1+\frac{2\eta}{3}\right)  \sum_{\substack{\ell=0\\ (1+\eta)^{\ell-1}\geq r_{k}}}^{L}\left(\frac{y_{ijk\ell}\cdot |I_{\ell}|}{d_{ijk}}(1+\eta)^{\ell-1}+\frac{1}{2}y_{ijk\ell}\cdot |I_{\ell}|\right)
\end{eqnarray*}
for the coflow scheduling problem with release times. We also can obtain
\begin{eqnarray*}\label{thm5:eq2}
E[C_{k}] & \leq  & 2 \left(1+\eta\right) \sum_{\substack{\ell=0\\ (1+\eta)^{\ell-1}\geq r_{k}}}^{L}\left(\frac{y_{ijk\ell}\cdot |I_{\ell}|}{d_{ijk}}(1+\eta)^{\ell-1}+\frac{1}{2}y_{ijk\ell}\cdot |I_{\ell}|\right)
\end{eqnarray*}
for the coflow scheduling problem without release times. When $\eta=\epsilon/2$ and $\epsilon>0$, this together with constraints (\ref{interval:d}) and inductive application of Lemma~\ref{lem:lem3} yields the theorem.
\end{proof}

\section{Concluding Remarks}\label{sec:Conclusion}
This paper investigates the scheduling problem of coflows with release times, aiming to minimize the total weighted completion time. Existing literature has predominantly concentrated on establishing the scheduling order of coflows. We advance this research by enhancing performance through the determination of the flow scheduling order. Our approximation algorithm achieves approximation ratios of $3$ and $2+\frac{1}{LB}$ for arbitrary and zero release times, respectively, where $LB$ is the minimum lower bound of coflow completion time. To further enhance time complexity, we reduce the number of variables in linear programming. Consequently, our improved approximation algorithm achieves approximation ratios of $3 + \epsilon$ and $2 + \epsilon$ for arbitrary and zero release times, respectively.

\end{document}